\documentclass[conference]{IEEEtran}
\IEEEoverridecommandlockouts           

\usepackage{setspace,cite}
\usepackage[overload]{empheq} 
\usepackage{graphicx}
\usepackage{epstopdf}
\usepackage{graphics} 
\usepackage{epsfig} 
\usepackage{color}
\usepackage{amsmath}
\usepackage{amsthm}
\usepackage{amssymb} 
\usepackage{verbatim}  
\usepackage{mathtools}
\usepackage{subfig}
\usepackage{tabularx}
\usepackage{amsthm}
\usepackage{cases}
\usepackage{verbatim}  
\usepackage{relsize}
\newtheorem{example}{Example}
\newtheorem{theorem}{Theorem}

\newtheorem{definition}{Definition}
\usepackage[ colorlinks = true,
linkcolor = blue,
urlcolor = blue,
citecolor = red,
anchorcolor = green,
]{hyperref}

\DeclareMathOperator\supp{supp}

\def\mcl{\mathcal}

\allowdisplaybreaks

\begin{document}
	%
	\title{ Robustness of Maximal $\alpha$-Leakage to Side Information}
	
	\author{\IEEEauthorblockN{Jiachun Liao, Lalitha Sankar, Oliver Kosut}
	\IEEEauthorblockA{School of Electrical, Computer and Energy Engineering,\\
		Arizona State University\\
		Email: \{jiachun.liao,lalithasankar,okosut\}@asu.edu}
	\and
	\IEEEauthorblockN{Flavio P. Calmon}
	\IEEEauthorblockA{School of Engineering and Applied Sciences\\
		Harvard University\\
		Email: flavio@seas.harvard.edu\\  
	}
\thanks{This material is based upon work supported by the National Science Foundation under Grant No. CCF\--1350914.}
}
	\maketitle %
	\begin{abstract}
	Maximal $\alpha$-leakage is a tunable measure of information leakage based on the accuracy of guessing an arbitrary function of private data based on public data. The parameter $\alpha$ determines the loss function used to measure the accuracy of a belief, ranging from log-loss at $\alpha=1$ to the probability of error at $\alpha=\infty$. To study the effect of side information on this measure, we introduce and define \textit{conditional maximal $\alpha$-leakage}. We show that, for a chosen mapping (channel) from the actual (viewed as private) data to the released (public) data and some side information, the conditional maximal $\alpha$-leakage is the supremum (over all side information) of the conditional Arimoto channel capacity where the conditioning is on the side information. We prove that if the side information is conditionally independent of the public data given the private data, the side information cannot increase the information leakage.
    \end{abstract}

%

\IEEEpeerreviewmaketitle

\section{Introduction}

The performance of security and privacy systems can vary significantly depending on the side information available to an adversary (see, for example, \cite{PowerAttack_kocher1999, Deanonymized_narayanan2008}). 
In general, it is difficult to account for the specific implementation intricacies of real-world privacy mechanisms when determining the  risk posed by adversarial side information. Instead, a more tractable approach is to asses the side-information resilience of the \emph{privacy (or information leakage) metric} used to design a given mechanism. Ideally, a privacy metric should quantify not only the risk incurred against an adversary that observes the output of the system, but also capture the robustness against different amounts of side information an adversary may have.

Despite the array of (often overlapping) privacy/information leakage measures  proposed over the past decade,  few   metrics  ensure robustness against side information. Differential privacy (DP)  
\cite{Dwork_DP_Survey}, for example, captures privacy in the context of querying statistical databases.  One of the key advantages of DP is that it is robust to arbitrary external knowledge (side information). This robustness is formalized in \cite{DPRobusttoSideInform_Kasiviswanathan15}, wherein the authors model side information by a prior probability distribution on the support of the original dataset. More recently, Issa et al. introduced maximal leakage (MaxL), which is essentially the maximal logarithmic gain in the probability of correctly guessing any arbitrary function of original data from released data \cite{OperationalLeak_issa2018}. The impact of side information on MaxL was studied in \cite{OperationalLeak_issa2018}. In our previous work, we introduced maximal $\alpha$-leakage as a measure of information leakage, which is proved to be related to Sibson and Arimoto mutual information, and incorporates MaxL \cite{TunableLeak_JL18}.

In this work, we characterize the robustness of maximal $\alpha$-leakage to side information. We model side information as a random variable observed by an adversary that is interested in learning an arbitrary function of the original data from the released data. We represent this as a conditional Markov chain, which is also used by Issa \textit{et al.} to study the effect of side information on maxL \cite[Def. 6]{OperationalLeak_issa2018}. 
This formulation naturally leads to the definition of conditional maximal $\alpha$-leakage, which is an extended version of maximal $\alpha$-leakage that accounts for side information. We demonstrate that
maximal $\alpha$-leakage upper bounds conditional maximal $\alpha$-leakage if the side information is conditionally independent of the released data given the original data.  That is, maximal $\alpha$-leakage is robust to arbitrary side information that is not used in generating the released data from the original data. This surprising result provides further motivation for using $\alpha$-leakage as a robust and tunable privacy metric. 
Finally, other metrics of note that may be amenable to such analysis include probability of correct guessing \cite{privacyGuessing_asoodeh2017}, total variation-based metrics \cite{TVDprivacy_Rassouli&Gunduz18}, and metrics based on  R{\'e}nyi divergence \cite{R-DP_Mironov17}.


\section{Preliminaries}\label{Sec:Preliminaries}

We begin by reviewing maximal $\alpha$-leakage. To do so, we first present Sibson MI \cite{alphaMI_Sibson1969} and Arimoto MI \cite{AlphaMI_Arimoto1975}. 
\begin{definition}
   Let $X$ and $Y$ be two discrete random variables where $(X,Y)\sim P_{XY}$, 
   The Sibson MI of order $\alpha\in(0,1)\cup(1,\infty)$ is given by
	\begin{align}
	\hspace{-5pt}I_\alpha^{\text{S}}(X;Y)
	\label{eq:Sibson_MI}
	= \frac{\alpha}{\alpha-1}\log \sum\limits_{y}\left(\sum\limits_{x}P_X(x)P_{Y|X}(y|x)^{\alpha}\right)^{\frac{1}{\alpha}}.
	\end{align}
	The Arimoto MI of order $\alpha\in(0,1)\cup(1,\infty)$ is defined as
	\begin{align}
	\label{eq:Def_ArimotoMI}
	I_\alpha^{\text{A}}(X;Y)&\triangleq H_{\alpha}(X)-H_{\alpha}(X|Y)\\
	\label{eq:Arimoto_MI}
	&=\frac{\alpha}{\alpha-1}\log\sum\limits_{y }\mathsmaller{\left(\frac{\sum\limits_{x}P_X(x)^{\alpha}P_{Y|X}(y|x)^{\alpha}}{\sum\limits_{x }P_X(x)^{\alpha}}\right)^{\frac{1}{\alpha}}},
	\end{align}
	where $H_{\alpha}$ is R{\'e}nyi entropy \cite{measures_renyi1961} and $H_{\alpha}(X|Y)$ is Arimoto conditional entropy of $X$ given $Y$ defined as
	\begin{align}
	\label{eq:Def_ArimotoConditionalEntropy}
	H_{\alpha}(X|Y)=
	\frac{\alpha}{1-\alpha}\log\sum\limits_{y}\mathsmaller{\left(\sum\limits_{x}\mathsmaller{P_X(x)^{\alpha}P_{Y|X}(y|x)^{\alpha}}\right)^{\frac{1}{\alpha}}}.
	\end{align}
	All of these quantities are defined by their continuous extensions for $\alpha=1$ or $\infty$. Note that both the Sibson and Arimoto MIs reduce to the Shannon MI at $\alpha=1$.
\end{definition}

Let $X$ and $Y$ represent the original (private) data and released (public) data, respectively, and let $U$ represent an arbitrary (potentially random) function of $X$ that the observer (a curious or malicious user with access to the released data $Y$) is interested in learning. 
We introduced maximal $\alpha$-leakage in \cite[Def. 5]{TunableLeak_JL18} to quantify an adversary's ability to infer \textit{any function} (ranging from the maximal likely realization to the posterior distribution) of data $X$ from the released $Y$. We review the definition below.
\begin{definition}[Maximal $\alpha$-Leakage]\label{Def:GeneralLeakge}
	Given a joint distribution $P_{X,Y}$ on finite alphabets $\mcl X\times\mcl Y$, the maximal $\alpha$-leakage, for $1\leq \alpha\leq \infty$, from $X$ to $Y$ is defined as
	\begin{align}
	&\mcl L_{\alpha}^{\text{max}}(X\to Y)\nonumber\\
	\label{eq:MaxAlplLeak_definition}
	\triangleq&\sup_{U- X- Y }\lim_{\alpha'\to \alpha} \frac{\alpha'}{\alpha'-1}\log\frac{\max\limits_{P_{\widehat{U}|Y}}\mathbb{E}\left[P_{\widehat{U}|Y}(U|Y)^{\frac{\alpha'-1}{\alpha'}}\right]}{\max\limits_{P_{\widehat{U}}}\mathbb{E}\left[P_{\widehat{U}}(U)^{\frac{\alpha'-1}{\alpha'}}\right]},	\\
	\label{eq:MaxAlplLeak_EquivDef}
	=&\begin{cases}
	\sup\limits_{P_{\tilde{X}}}I^{\text{S}}_\alpha(\tilde{X};Y)=\sup\limits_{P_{\tilde{X}}}I_{\alpha}^{\text{A}}(\tilde{X};Y),& 1<\alpha\leq \infty  \\
	I(X;Y),  &\alpha=1  
	\end{cases},
	\end{align}
	where (i) in \eqref{eq:MaxAlplLeak_definition} $U$ represents any function of $X$ and takes values from an arbitrary alphabet, and the objective function is defined as the $\alpha$-leakage from $U$ to $Y$; (ii) $I_{\alpha}^{\text{S}}$ and $I_{\alpha}^{\text{A}}$ in \eqref{eq:MaxAlplLeak_EquivDef} indicate Sibson and Arimoto MIs, respectively\cite{alphaMI_Sibson1969,AlphaMI_Arimoto1975}.
\end{definition}
Note that the optimal $P_{\hat{U}|Y}^*$ of the maximization in the numerator of the logarithmic term in \eqref{eq:MaxAlplLeak_definition} minimizes the expectation of the following \textit{$\alpha$-loss} function
\begin{equation}\label{poly_loss}
\ell_{\alpha}(u,y,P_{\hat{U}|Y})=\frac{\alpha}{\alpha-1} \big(1-P_{\hat{U}|Y}(u|y)^{\frac{\alpha-1}{\alpha}}\big),
\end{equation}
for each $\alpha\in(1,\infty)$. The limit of the loss function in \eqref{poly_loss} leads to the 0-1 loss (for $\alpha=1$) and the probability of (guessing) error (for $\alpha=\infty$) functions, respectively. 
Consequently, for $\alpha=1$ and $\infty$, maximal $\alpha$-leakage simplifies to MI and MaxL, respectively. For $\alpha> 1$, maximal $\alpha$-leakage is essentially the Arimoto channel capacity (with a support-set constrained input distribution) \cite{AlphaMI_Arimoto1975}.

\section{Conditional Tunable Information Leakage Measures}\label{Sec:Information Leakage Measures}	
Given a pair of original and released data $(X,Y)$, let $Z$ be the knowledge of some particular adversary or third-party about $(X,Y)$. Before introducing the conditional maximal $\alpha$-leakage, we introduce the following simpler measure, the conditional $\alpha$-leakage. Here, the adversary is interested only in guessing $X$, rather than a function of $X$.
    
	\begin{definition}[{Conditional $\alpha$-Leakage}]\label{Def:CondAlphaLeak_Def}
	Given a joint distribution $P_{XYZ}$ and an estimator $\hat{X}$ with the same support as $X$, the $\alpha$-leakage from $X$ to $Y$ given $Z$ is defined as
	\begin{align}
	&\mcl L_{\alpha}(X \to  Y|Z)\nonumber\\
	\triangleq\,\,&\frac{\alpha}{\alpha-1}\log\frac{\max\limits_{P_{\hat{X}|Y,Z}}\mathbb{E}\left[P_{\hat{X}|Y,Z}(X|Y,Z)^{\frac{\alpha-1}{\alpha}}\right]}{\max\limits_{P_{\hat{X}|Z}}\mathbb{E}\left[P_{\hat{X|Z}}(X|Z)^{\frac{\alpha-1}{\alpha}}\right]}
		\label{eq:CondAlphaLeak_Def}
	\end{align}
	for $1<\alpha<\infty$ and by the continuous extension of \eqref{eq:CondAlphaLeak_Def} for $\alpha = 1$ and $\infty$.
   \end{definition}
The conditional $\alpha$-leakage quantifies the maximal logarithmic gain in inferring various information about $X$ when an adversary \textit{with arbitrary side information $Z$} has access to $Y$. To understand the effect of the side information $Z$ on leakage about \textit{any function} $U$ of $X$ through $Y$, we define conditional maximal $\alpha$-leakage as follows.
	\begin{definition}[Maximal Conditional $\alpha$-Leakage]\label{Def:CondMaxAlphaLeak_Def}
		Given a joint distribution $P_{XYZ}$, 
		for $1\leq\alpha\leq\infty$, the conditional maximal $\alpha$-leakage from $X$ to $Y$ given $Z$ is defined as
		\begin{align}
		 \mcl L_{\alpha}^{\text{max}}(X\to Y|Z) 	
		\triangleq\, \sup_{U:U- X- Y|Z }\mcl L_{\alpha}(U \to  Y|Z)
		\label{eq:CondMaxAlphaLeak_Def}
		\end{align}
		 where $U$ represents any function of $X$ and takes values from an arbitrary alphabet.
		 Moreover, the expression $U-X-Y|Z$ represents the conditional Markov chain constraint where
		 	\begin{align}
		 	P_{UXY|Z}(uxy|z)
		 	=P(x|z)P(u|xz)P(y|xz)\label{eq:CondMarkovChain}.
		 	\end{align}
		 	 Therefore, the conditional Markov chain $U-X-Y|Z$ is equivalent to $U-(X,Z)-Y$.
	\end{definition}
Note that conditional maximal $\alpha$-leakage takes side information $Z$ into consideration via the conditional Markov chain $U-X-Y|Z$, which is equivalent to $U-(X,Z)-Y$. Therefore, conditional maximal $\alpha$-leakage is designed under the two assumptions: side information $Z$ can be arbitrarily related to $X$ and $U$, and the released data $Y$ will not provide more information about $U$ than $X$ and $Z$.

The Markov chain $U-X-Y$ models inferences for a function $U$ of $X$ from $Y$. To involve side information $Z$ in the inferences, 
beyond the conditional Markov chain in Def. \ref{Def:CondMaxAlphaLeak_Def}, there are two other possibilities: 
\begin{itemize}
		\item[(i)] If the side information $Z$ that an adversary has is arbitrarily related to the function of interest $U$, but conditionally independent of released data $Y$ given $X$, we have $(U,Z)-X-Y$. For example, if $X$ is an individual's public records \textit{without} voter registration indicated by $Z$ and $Y$ is a noisy release of $X$, then when $U$ is the political preference of this person, $Z$ can provide extra information about $U$ and is conditionally independent of $Y$ given $X$. 
   
		\item[(ii)] If the side information $Z$ does not provide more information about the function of interest $U$ than original data $X$ does, but can be arbitrarily related to the released data $Y$, we have $U-X-(Y,Z)$. For example, if $X$ is an individual's public records \textit{with} voter registration, $Z$ is a noisy release of the voter registration in $X$, and $Y$ is an update of $Z$, then when $U$ is the political preference of this person, $Z$ cannot provide extra information about $U$ than $X$ does but it can be helpful in inferring $U$ from $Y$ (i.e., the Markov chain $Z-X-Y$ does not hold).
\end{itemize}
Note that in either Markov chain mentioned above, $U$ and $Y$ are conditionally independent given $X$ and $Z$. In this sense, the proposed conditional Markov chain generally models side information in privacy-protection problems.

\section{Main Results}
In this section, we explore the effect of side information on inferring any function of original data from released data. First, we simplify the expression of conditional maximal $\alpha$-leakage, and then, compare leakages of a privacy mechanism measured by conditional maximal $\alpha$-leakage and maximal $\alpha$-leakage.
	
The following theorem simplifies the expression of the conditional $\alpha$- leakage in \eqref{eq:CondAlphaLeak_Def} as a conditional Arimoto MI based on Arimoto conditional entropy.
	\begin{definition}
		Given a joint distribution $P_{X,Y,Z}$, the conditional Arimoto mutual information, for $1\leq\alpha\leq \infty$, between $X$ and $Y$ given $Z$ is defined as
		\begin{align}\label{eq:ConAlphaMI}
			I_{\alpha}^{\text{A}}(X;Y|Z)
			\triangleq
			H_{\alpha}^{\text{A}}(X|Z)-H_{\alpha}^{\text{A}}(X|YZ)
		\end{align}
		where $H_{\alpha}^{\text{A}}(\cdot|\cdot)$ indicates Arimoto conditional entropy.
	\end{definition}
Note that for $\alpha=1$, the conditional Arimoto MI in \eqref{eq:ConAlphaMI} is exactly the conditional Shannon MI $I(X;Y|Z)$.	
	\begin{theorem}\label{Thm:CondAlphaLeak_Equialent}
		For $\alpha\in[1,\infty]$, conditional $\alpha$-leakage defined in \eqref{eq:CondAlphaLeak_Def} simplifies to
		\begin{align}
		\mcl L_{\alpha}(X\to Y|Z)
		= I_{\alpha}^{\text{A}}(X;Y|Z)\label{eq:CondAlphaLeak_EqConAlphaMI}.
		\end{align}	
	\end{theorem}
The proof hinges on solving the two optimal problems in \eqref{eq:CondAlphaLeak_Def} by using Karush-–Kuhn-–Tucker (KKT) conditions. 
As this proof is nearly identical to that of \cite[Thm. 1]{TunableLeak_JL18}, we omit it.

Based on the result of Thm. \ref{Thm:CondAlphaLeak_Equialent}, we obtain a simplified expression for conditional maximal $\alpha$-leakage. Specifically, the simplified expression for $\alpha>1$ is related to a variant of the Sibson MI defined as follows.
\begin{definition}\label{Def:S-MICondEvent}
	Let $P_{X,Y|Z=z}$ indicate a conditional joint distribution of $X,Y$ given an event $Z=z$. The event-conditional Sibson MI between $X$ and $Y$ given $Z=z$ is defined 
	\begin{IEEEeqnarray}{l}\label{eq:Def_S-MICondEvent}
		\hspace{-10pt}I_{\alpha}^{\text{S}}(X;Y|Z\hspace{-2pt}=\hspace{-2pt}z)\hspace{-2pt}=\hspace{-2pt}\frac{\alpha}{\alpha-1}\hspace{-2pt}\log \sum\limits_{y}\hspace{-2pt}\left(\hspace{-2pt}\sum\limits_{x} P(x|z) P(y|x,z)^\alpha\hspace{-2pt}\right)^{\hspace{-2pt}\frac{1}{\alpha}}
	\end{IEEEeqnarray}
for $1<\alpha<\infty$ and by the continuous extension of \eqref{eq:Def_S-MICondEvent} for $\alpha = 1$ and $\infty$.
\end{definition}

\begin{theorem}\label{Thm:MaxCondAlphaLeak_Equialent}
	For $\alpha\in[1,\infty]$, the conditional maximal $\alpha$-leakage defined in \eqref{eq:CondMaxAlphaLeak_Def} simplifies to
	\begin{align}\label{eq:CondMaxAlphaLeak_Equialent}
	\hspace{-10pt}&\mcl L_{\alpha}^{\text{max}}(X\to Y|Z)\nonumber\\
	\hspace{-10pt}=&\begin{cases}
	\sup\limits_{z\in \supp(Z)}\, \sup\limits_{\substack{P_{\tilde{X}|Z=z}\\ \ll P_{X|Z=z}}}\,I_{\alpha}^{\text{S}}(\tilde{X};Y|Z=z),& \alpha\in(1,\infty]\\
	I(X;Y|Z),&\alpha=1.
	\end{cases}
	\end{align}	
	where $\supp(Z)$ indicates the support of $Z$ and $I_{\alpha}^{\text{S}}(X;Y|Z=z)$ is defined in \eqref{eq:Def_S-MICondEvent}. 
\end{theorem}
A detailed proof is in the Appendix. Note that given a channel $P_{Y|X}$, $\sup_{X} I_{\alpha}^{\text{S}}(X;Y)=\sup_{X} I_{\alpha}^{\text{A}}(X;Y)$ for $1\leq\alpha\leq \infty$, and the quantity is called Arimoto channel capacity \cite{RenyiandHypothesisTest_Csiszar1995,alphaMI_verdu}. Thus, for $\alpha>1$, conditional maximal $\alpha$-leakage is the maximal conditional Arimoto channel capacity of channels (from $X$ to $Y$) where the channel state is controlled by $Z$.

The following theorem shows a relationship between conditional maximal $\alpha$-leakage and maximal $\alpha$-leakage.
\begin{theorem}\label{Thm:CondionReduceMaxAlphaLeak}
	For conditional maximal $\alpha$-leakage defined in \eqref{eq:CondMaxAlphaLeak_Def}, if $Z-X-Y$ holds, then
	\begin{align}\label{eq:CondionReduceMaxAlphaLeak}
	\mcl L_{\alpha}^{\text{max}}(X\to Y|Z)\leq \mcl L_{\alpha}^{\text{max}}(X\to Y).
	\end{align}	
\end{theorem}
\begin{proof}
From Thm. \ref{Thm:MaxCondAlphaLeak_Equialent}, we have that for $\alpha>1$
\begin{IEEEeqnarray}{r l}
	&\mcl L_{\alpha}^{\text{max}}(X\to Y|Z)\nonumber\\
	=&	\sup\limits_z \hspace{-3pt}\sup\limits_{P_{\tilde{X}}\ll P_{X|Z=z}}\hspace{-3pt} \frac{\alpha}{\alpha-1}\log \sum\limits_{y}\hspace{-2pt}\left(\hspace{-1pt}\sum\limits_{x} P_{\tilde{X}}(x) P_{Y|X}(y|x)^\alpha\hspace{-2pt}\right)^{\frac{1}{\alpha}}\label{eq:CondRedLeak_InPf1}\\
	\leq &  \sup\limits_{P_{\tilde{X}}\ll P_{X}} \frac{\alpha}{\alpha-1}\log \sum\limits_{y}\left(\sum\limits_{x} P_{\tilde{X}}(x) P_{Y|X}(y|x)^\alpha\right)^{\frac{1}{\alpha}}\label{eq:CondRedLeak_InPf2}\\
	=&\mcl L_{\alpha}^{\text{max}}(X\to Y)\label{eq:CondRedLeak_InPf4}
\end{IEEEeqnarray}
where \eqref{eq:CondRedLeak_InPf1} holds because the Markov chain $Z-X-Y$ allows us to replace $P_{Y|X,Z}$ with $P_{Y|X}$; and the inequality in \eqref{eq:CondRedLeak_InPf2} is from the fact that for any $z$ conditioning on $Z$ can only reduce the support of $X$;
and the equality in \eqref{eq:CondRedLeak_InPf4} is from Thm. 2 in \cite{TunableLeak_JL18}. 
For $\alpha=1$, from Thm. \ref{Thm:MaxCondAlphaLeak_Equialent} we have
\begin{align}
\mcl L_{1}^{\text{max}}(X\to Y|Z)=I(X;Y|Z),
\end{align} such that if $Z-X-Y$ holds, 
\begin{align}
I(X;Y|Z)\leq I(X;Y)=\mcl L_{1}^{\text{max}}(X\to Y),
\end{align} where the inequality and equality are from \cite[Sec. 2.8]{IT_Cover} and \cite[Thm. 2]{TunableLeak_JL18}, respectively. Therefore, for $Z-X-Y$, $\mcl L_{\alpha}^{\text{max}}(X\to Y|Z)\leq \mcl L_{\alpha}^{\text{max}}(X\to Y)$. 
\end{proof}	
Thm. \ref{Thm:CondionReduceMaxAlphaLeak} shows that if side information ($Z$) and released data ($Y$) is conditionally independent on the original data ($X$), the amount of information that $Y$ can leak about $X$ will not increase. That is, if side information is not involved in generating the released data from the original data, in terms of maximal $\alpha$-leakage, it will not help an adversary get more information about the original data from the released data. Therefore, the (unconditional) maximal $\alpha$-leakage represents a bound not only on the amount of information leaked in $Y$ about an arbitrary function of $X$, but it is also a bound on the amount of information leaked in $Y$ about $X$ to an adversary with \emph{arbitrary side-information}, provided $Y$ is generated from $X$ using only private randomness. This gives significant new meaning to the maximal $\alpha$-leakage. The following example illustrates the result in Thm. \ref{Thm:CondionReduceMaxAlphaLeak}.
	\begin{example}\label{exmp:binary-condReduceLeak}
		Let the original data $X$ uniformly take values from the binary alphabet $\{0,1\}$, and the released data $Y$ be generated by a binary symmetric channel with a crossover probability $0<p<0.5$. 
		Here, the maximal $\alpha$-leakage from $X$ to $Y$ is 
		\begin{align}\label{eq:MaxAlphaLeak-Example}
		&\mcl L^{\text{max}}_{\alpha}(X\to Y)\nonumber\\
		=&\begin{cases}
		\log 2+\frac{1}{\alpha-1}\log \left(p^{\alpha}+(1-p)^\alpha\right),&\alpha>1\\
		\log 2-H(p),&\alpha=1
		\end{cases}
		\end{align}
		where $H(p)=-p\log p-(1-p)\log (1-p)$.
		Let the side information $Z\in \{0,1\}$ be generated from $X$ via a binary symmetric channel with a crossover probability $0\leq q\leq 0.5$, such that $Z-X-Y$ holds. From Thm. \ref{Thm:MaxCondAlphaLeak_Equialent}, we know that $\mcl L^{\text{max}}_{\alpha}(X\to Y|Z)=0$ for $q=0$, and if $q\neq 0$
		\begin{align}
		&\mcl L^{\text{max}}_{\alpha}(X\to Y|Z)\nonumber\\
		=&\begin{cases}
		\log 2+\frac{1}{\alpha-1}\log \left(p^{\alpha}+(1-p)^\alpha\right),&\alpha>1\\
		H(p+q-2pq)-H(p),&\alpha=1.
		\end{cases}
		\end{align}
		Therefore, $\mcl L^{\text{max}}_{\alpha}(X\to Y|Z)\leq \mcl L^{\text{max}}_{\alpha}(X\to Y)$ with equality if and only if $\alpha>1$ or $q=0.5$. 
\end{example}
As a contrast, for the same binary $(X,Y)$ in Example \ref{exmp:binary-condReduceLeak} we show a case in which the Markov chain $Z-X-Y$ does not hold, so that side information causes the released data leak more information about the original data. 
\begin{example}
Let side information $Z \sim Ber(p)$ and $Z \perp X$.
Let $Y=X$ for $Z=0$ and $Y=X\oplus 1$ for $Z=1$, such that $P_{X,Y}$ is the same as that in Example \ref{exmp:binary-condReduceLeak}. From Thm. \ref{Thm:MaxCondAlphaLeak_Equialent}, we have $\mcl L^{\text{max}}_{\alpha}(X\to Y|Z)=\log 2 > \mcl L^{\text{max}}_{\alpha}(X\to Y)$ from \eqref{eq:MaxAlphaLeak-Example}.
\end{example}

\section{Concluding Remarks}
We have shown that in a data publishing setting, when the released data is generated from original data via private randomness (i.e., side information is not involved in the generation), maximal $\alpha$-leakage is robust to \textit{arbitrary side information} an adversary may have. Building upon our earlier result on composition that leakage over multiple releases can be bounded as $\mcl L^{\text{max}}_{\alpha}(X\hspace{-2pt}\to\hspace{-2pt}(Y,Z))\leq \mcl L^{\text{max}}_{\alpha}(X\hspace{-2pt}\to\hspace{-2pt} Y)+\mcl L^{\text{max}}_{\alpha}(X\hspace{-2pt}\to\hspace{-2pt} Z)$, we conjecture a tighter composition theorem that limits leakage over multiple releases as
\begin{IEEEeqnarray}{l}\label{eq:MaxAlphaLK-ChainRule}
\hspace{-12pt}\mcl L^{\text{max}}_{\alpha}(X\hspace{-2pt}\to\hspace{-2pt}(Y,Z))\leq \mcl L^{\text{max}}_{\alpha}(X\hspace{-2pt}\to\hspace{-2pt} Y)+\mcl L^{\text{max}}_{\alpha}(X\hspace{-2pt}\to\hspace{-2pt} Z|Y),
\end{IEEEeqnarray}thereby suggesting that successive releases should leverage adversarial knowledge. 

\appendices
\section*{Appendix: Proof of Theorem \ref{Thm:MaxCondAlphaLeak_Equialent} }\label{Proof:Thm:MaxCondAlphaLeak_Equialent}
	From Thm. \ref{Thm:CondAlphaLeak_Equialent}, we can simplify $\mcl L_{\alpha}^{\text{max}}(X\to Y|Z)$ in \eqref{eq:CondMaxAlphaLeak_Def} as
	\begin{align}
	\mcl L_{\alpha}^{\text{max}}(X\to Y|Z)=
	\sup\limits_{U- X- Y|Z} I_{\alpha}^{\text{A}}(U;Y|Z).
	\end{align}
	For $\alpha=1$, we have
	\begin{align}
	\mcl L_{1}^{\text{max}}(X\to Y|Z)=\sup_{U:U-(X,Z)-Y} I(U;Y|Z).
	\end{align}
	Under the Markov chain $U-X-Y|Z$, by the data processing inequality, we have $I(U;Y|Z)\leq I(X;Y|Z)$ with equality if and only if $I(X;Y|U,Z)=0$.
	Thus,
	\begin{align}
	\mcl L_{1}^{\text{max}}(X\to Y|Z)=I(X;Y|Z).
	\end{align}
	
Now consider $\alpha>1$. We first upper bound $\mcl L_\alpha^{\max}(X\to Y|Z)$. To show that this is upper bounded by the expression in \eqref{eq:CondMaxAlphaLeak_Equialent}, we show that for any variable $U$ satisfying the Markov chain $U-X-Y|Z$, the conditional $\alpha$-leakage is upper bounded by this same expression. For any such $U$, we have
\begin{IEEEeqnarray}{l l}
		&I_{\alpha}^{\text{A}}(U;Y|Z)\nonumber\\
	= &\frac{\alpha}{\alpha-1}\log  \frac{\sum\limits_{y,z}\left(\sum\limits_u P_{U,Y,Z}(u,y,z)^\alpha\right)^{\frac{1}{\alpha}}}{\sum\limits_z\left(\sum\limits_uP_{U,Z}(u,z)^\alpha\right)^{\frac{1}{\alpha}}}\label{eq:Thm-CMALk_Eq-InPf1}\\
	\leq&\frac{\alpha}{\alpha-1}\log \sup_{z\in\supp(Z)} \frac{\sum\limits_{y}\left(\sum\limits_u P_{U,Y,Z}(u,y,z)^\alpha\right)^{\frac{1}{\alpha}}}{\left(\sum\limits_uP_{U,Z}(u,z)^\alpha\right)^{\frac{1}{\alpha}}}
	\label{eq:Thm-CMALk_Eq-InPf2}\\
	=& \sup_{z\in\supp(Z)} I_\alpha^{\text{A}}(U;Y|Z=z)   \label{eq:Thm-CMALk_Eq-InPf3} \\	
	\le& \sup_{z\in\supp(Z)} \sup_{P_{\tilde{X}|\tilde{U}}:P_{\tilde{X}|\tilde{U}}\ll P_{X|Z=z}}  \sup_{P_{\tilde{U}}} I_{\alpha}^{\text{A}}(\tilde{U};Y|Z=z)   \label{eq:Thm-CMALk_Eq-InPf4} \\
   = & \sup_{z\in\supp(Z)} \sup_{P_{\tilde{X}|\tilde{U}}:P_{\tilde{X}|\tilde{U}}\ll P_{X|Z=z}}  \sup_{P_{\tilde{U}}} I_\alpha^{\text{S}}(\tilde{U};Y|Z=z)   \label{eq:Thm-CMALk_Eq-InPf5} \\		
		\le& \sup_{z\in\supp(Z)} \sup_{P_{\tilde{X}}\ll P_{X|Z=z}} I_\alpha^{\text{S}}(\tilde{X};Y|Z=z)   \label{eq:Thm-CMALk_Eq-InPf6} 
\end{IEEEeqnarray}
	where 
	\begin{itemize}
		\item  the inequality in \eqref{eq:Thm-CMALk_Eq-InPf2} is from the fact that 
		for any nonnegative $a_i,b_i$, 
		\begin{equation}
		\frac{\sum_i a_i}{\sum_i b_i} \le \max_i \frac{a_i}{b_i},
		\end{equation}
		\item \eqref{eq:Thm-CMALk_Eq-InPf3} follows by the definition of Arimito MI,
		\item in \eqref{eq:Thm-CMALk_Eq-InPf4}, the variables are distributed according to $P_{\tilde{U}}(u)P_{\tilde{X}|\tilde{U}}(x|u) P_{Y|X,Z}(y|x,z)$,
		\item \eqref{eq:Thm-CMALk_Eq-InPf5} follows because Arimoto and Sibson MIs have the same supremum over the input distribution,
		\item \eqref{eq:Thm-CMALk_Eq-InPf6} follows from the facts that Sibson MI satisfies the data processing inequality, and $\tilde{U}-\tilde{X}-Y|Z=z$ forms a Markov chain.
   \end{itemize} 
We now lower bound $\mcl L_\alpha^{\max}(X;Y|Z)$ by constructing a specific $U$ satisfying $U-X-Y|Z$.
	For a given $P_{X,Y,Z}$, let 
	\begin{align}\label{eq:Thm-CMALk_LBDefz-InPf1}
	\hspace{-10pt}z^*\hspace{-2pt}=\arg \hspace{-4pt}\sup_{\substack{z\in\supp(Z)}} \hspace{-6pt}\sup_{\substack{P_{\tilde{X}}\\\ll P_{X|Z=z}}}\hspace{-6pt} 
	\sum\limits_{y}\hspace{-2pt}\left(\hspace{-2pt}\sum\limits_{x}\hspace{-1pt}P_{\tilde{X}}(x) P_{Y|X,Z}(y|x,z)^\alpha  \hspace{-3pt}\right)^{\hspace{-4pt}\frac{1}{\alpha}}\hspace{-7pt}.
	\end{align}
	We will define a variable $U$ with alphabet consisting of several disjoint subsets. We use $\mcl X_{z^*}$ to indicate the conditional support of $X$ given $Z=z^*$, i.e., $\mcl X_{z^*}\triangleq \{x\in\mcl X: P_{X,Z}(x,z^*)>0\}$. For each $x\in\mathcal{X}_{z^\star}$, let $\mathcal{U}_{x,z^\star}$ be disjoint, finite sets. Also let $\mathcal{U}_0$ be a finite set (disjoint from those above). The cardinality of each of these sets will be determined later. Finally, let the alphabet of $U$ be $	\mathcal{U}=\mcl U_0\cup \bigcup_{x\in\mcl X_{z^\star}} \mathcal{U}_{x,z^\star}.$
	We define the conditional distribution $P_{U|X,Z}$ as follows. Let
	\begin{align}\label{eq:Thm-CMALk_ConU-Inp}
		P_{U|X,Z}(u|x,z)=\begin{cases}
		\frac{1}{|\mathcal{U}_{x,z^\star}|}, & z=z^\star,\ u\in\mathcal{U}_{x,z^\star} \\
		\frac{1}{|\mathcal{U}_0|}, & z\ne z^\star,\ u\in\mathcal{U}_0\\
		0, & \text{otherwise}.
		\end{cases}
	\end{align}
	For the constructed $U$ above, the conditional Arimoto MI is
	\begin{IEEEeqnarray}{l }
	\hspace{-15pt}	I_{\alpha}^{\text{A}}(U;Y|Z) 
		= \frac{\alpha}{\alpha-1}\log \frac{\sum\limits_{y,z}\left(\sum\limits_{u} P_{U,Y,Z}(u,y,z)^\alpha\right)^{\frac{1}{\alpha}}}{\sum\limits_z\left(\sum\limits_{u}P_{U,Z}(u,z)^\alpha\right)^{\frac{1}{\alpha}}}.\label{eq:Thm-CMALk_ConUeq-InPf0}
	\end{IEEEeqnarray}
The numerator in \eqref{eq:Thm-CMALk_ConUeq-InPf0} can be written as
\begin{IEEEeqnarray}{l l}
	&\sum_{y,z}\left( \sum_u P_{U,Y,Z}(u,y,z)^\alpha\right)^{\frac{1}{\alpha}}\nonumber\\
	=&\sum_{y,z} \left(\sum_u \left(\sum_x P_{U|X,Z}(u|x,z) P_{X,Y,Z}(x,y,z)\right)^\alpha \right)^{\frac{1}{\alpha}}\\	
	=&\sum_{y,z\ne z^*} \left( |\mathcal{U}_0| \left(\sum_x \frac{1}{|\mathcal{U}_0|} P_{X,Y,Z}(x,y,z)\right)^\alpha \right)^{\frac{1}{\alpha}}\nonumber\\
		& +\sum_y \left( \sum_x |\mathcal{U}_{x,z^*}| \left( \frac{1}{|\mathcal{U}_{x,z^*}|} P_{X,Y,Z}(x,y,z^*)\right)^\alpha\right)^{\frac{1}{\alpha}}\\	
	=& \frac{1-P_Z(z^*)}{|\mathcal{U}_0|^{1-\frac{1}{\alpha}}}
		+\sum_y \left(\sum_x |\mathcal{U}_{x,z^*}|^{1-\alpha} P_{X,Y,Z}(x,y,z^\star)^\alpha \right)^{\frac{1}{\alpha}} \label{eq:Thm-CMALk_ConUeq-InPf2}		
\end{IEEEeqnarray}
where the simplification in \eqref{eq:Thm-CMALk_ConUeq-InPf2} is from \eqref{eq:Thm-CMALk_ConU-Inp}.
A similar derivation for the denominator in \eqref{eq:Thm-CMALk_ConUeq-InPf0} gives 
\begin{align}
		&\sum_{z}\left( \sum_u P_{U,Z}(u,z)^\alpha\right)^{\frac{1}{\alpha}}\nonumber\\
		=& \frac{1-P(z^*)}{|\mathcal{U}_0|^{1-\frac{1}{\alpha}}}
		+\left(\sum_x |\mathcal{U}_{x,z^*}|^{1-\alpha} P_{X,Z}(x,z^\star)^\alpha \right)^{\frac{1}{\alpha}}. \label{eq:Thm-CMALk_ConUeq-InPf2D}
\end{align}
Note that for $\alpha>1$, as $|{\mcl U}_0|\to \infty$, $(1-P_Z(z^*))\frac{1}{|{\mcl U}_0|^{1-\frac{1}{\alpha}}} \to 0$. Therefore, for $\alpha>1$ we have
\begin{IEEEeqnarray}{l l}
	&\mcl L_{\alpha}^{\text{max}}(X\to Y|Z)\nonumber\\
	\geq & \,\frac{\alpha}{\alpha-1} \log \frac{\sum\limits_y \left(\sum\limits_x |\mathcal{U}_{x,z^*}|^{1-\alpha} P_{X,Y,Z}(x,y,z^\star)^\alpha \right)^{\frac{1}{\alpha}} }{\left(\sum\limits_x |\mathcal{U}_{x,z^*}|^{1-\alpha} P_{X,Z}(x,z^\star)^\alpha \right)^{\frac{1}{\alpha}}}\label{eq:Thm-CMALk_LB-InPf0}\\
	= & \,\frac{\alpha}{\alpha-1} \hspace{-2pt}\log\sum\limits_{y}\hspace{-3pt}\left(\hspace{-2pt} \frac{ \sum\limits_{x\in\mcl X_{z^*}}\hspace{-5pt}P_{Y|X,Z}(y|x,z^*)^\alpha P(x,z^*)^\alpha|\mcl U_{x,z^*}|^{1-\alpha} }{\sum\limits_{x'\in\mcl X_{z^*}} \hspace{-8pt}P(x',z^*)^\alpha|\mcl U_{x',z^*}|^{1-\alpha}} \hspace{-2pt}\right)^{\frac{1}{\alpha}}.
\end{IEEEeqnarray}
Let $\tilde{X}\in \mcl X_{z^*}$ be random variable with a distribution $P_{\tilde{X}}(x)=\frac{P_{X,Z}(x,z^*)^\alpha |\mcl U_{x,z^*}|^{1-\alpha} }{\sum_{x'\in\mcl X_{z^*}} P_{X,Z}(x',z^*)^\alpha|\mcl U_{x',z^*}|^{1-\alpha}}$. By properly choosing cardinalities $|\mcl U_{x,z^*}|$, for $x\in \mcl X_{z^*}$, we can approach an arbitrary distribution $P_{\tilde{X}}$ on the support $\mcl X_{z^*}$. In addition, the lower bound in \eqref{eq:Thm-CMALk_LB-InPf0} holds for any arbitrary choice of these cardinalities. 
Therefore, we have
\begin{IEEEeqnarray}{l l}
	\hspace{-10pt}&\mcl L_{\alpha}^{\text{max}}(X\to Y|Z)\nonumber\\
	\geq & \hspace{-3pt}\sup_{\substack{P_{\tilde{X}}\\\ll P_{X|Z=z^*}}}\hspace{-2pt}\frac{\alpha}{\alpha-1}\hspace{-2pt}\log\hspace{-2pt} \sum\limits_{y}\hspace{-3pt}\left(\hspace{-2pt}\sum\limits_{x}\hspace{-2pt}P_{\tilde{X}}(x) P_{Y|X,Z}(y|x,z^*\hspace{-1pt})^{\hspace{-1pt}\alpha } \hspace{-4pt}\right)^{\hspace{-2pt}\frac{1}{\alpha}}\\
		=& \hspace{-3pt}\sup_{\substack{z\in\\\supp(Z)}} \hspace{-2pt}\sup_{\substack{P_{\tilde{X}}\ll\\ P_{X|Z=z}}} \hspace{-3pt}\frac{\alpha}{\alpha\hspace{-2pt}-\hspace{-2pt}1}\hspace{-2pt}\log  
		 \hspace{-2pt}\sum\limits_{y}\hspace{-3pt}\left(\hspace{-2pt}\sum\limits_{x}\hspace{-2pt} P_{\tilde{X}}(x) P_{Y\hspace{-1pt}|X\hspace{-0.5pt},\hspace{-0.5pt}Z}(y|x,\hspace{-1pt}z)^\alpha\hspace{-4pt}\right)^{\hspace{-4pt}\frac{1}{\alpha}}\label{eq:Thm-CMALk_LB-InPf2}
\end{IEEEeqnarray}
where \eqref{eq:Thm-CMALk_LB-InPf2} is from the definition of $z^*$ in \eqref{eq:Thm-CMALk_LBDefz-InPf1}.
From \eqref{eq:Thm-CMALk_Eq-InPf6} and \eqref{eq:Thm-CMALk_LB-InPf2}, we have that for $\alpha>1$
\begin{align}
\hspace{-6pt}\mcl L_{\alpha}^{\text{max}}(X\hspace{-2pt}\to\hspace{-1pt} Y|Z)=	\sup_{z\in\supp(Z)} \sup_{P_{\tilde{X}}\ll P_{X|Z=z}}\hspace{-8pt} I_\alpha^{\text{S}}(\tilde{X};Y|Z=z).
\end{align}

%
%

	\ifCLASSOPTIONcaptionsoff
	\newpage
	\fi
	
	\bibliographystyle{IEEEtran}
	\bibliography{JL_References}

\end{document}